\documentclass[12pt]{article}

\usepackage{amssymb,amsmath,amsfonts,eurosym,geometry,ulem,graphicx,caption,color,setspace,sectsty,comment,footmisc,caption,natbib,pdflscape,subfigure,array}
\usepackage{bbm}
\usepackage{color}  
\usepackage[hidelinks,colorlinks=true,linkcolor=blue,citecolor=blue]{hyperref}

\usepackage{natbib}
\usepackage{tikz}
\usepackage{pgfplots}

\pgfplotsset{width=6cm,compat=1.9}
\usepgfplotslibrary{fillbetween}

\newcolumntype{L}[1]{>{\raggedright\let\newline\\arraybackslash\hspace{0pt}}m{#1}}
\newcolumntype{C}[1]{>{\centering\let\newline\\arraybackslash\hspace{0pt}}m{#1}}
\newcolumntype{R}[1]{>{\raggedleft\let\newline\\arraybackslash\hspace{0pt}}m{#1}}

\usepackage[T1]{fontenc}
\usepackage[utf8]{inputenc}
\usepackage{setspace}
\usepackage{babel}
\usepackage{blindtext}
\usepackage[utf8]{inputenc}
\usepackage{amssymb}
\usepackage{amsthm}
\usepackage{amsmath}
\newtheorem{theorem}{Theorem}
\newtheorem{corollary}{Corollary}

\newtheorem{proposition}{Proposition}

\newtheorem{remark}{Remark}
\newtheorem{lemma}{Lemma}

\usepackage{array}
\usepackage{footmisc}

\title{Information-Robust Optimal Auctions}
\author{Wanchang Zhang\thanks{Department of Economics, University of California, San Diego.  Email: waz024@ucsd.edu. I thank    Songzi Du and Joel Sobel for helpful comments.}
}
\date{Current Draft: May 8, 2022\\
First Draft: April 20, 2022}

\begin{document}
\maketitle

\begin{abstract}
A single unit of a  good is sold to one of two bidders. Each bidder has either a high prior valuation or a low prior valuation for the good. Their prior valuations are independently and identically distributed.  Each bidder may observe an independently and identically distributed signal about her prior valuation. The seller knows the distribution of the prior valuation profile and knows that signals are independently and identically distributed, but does not know the signal distribution. In addition, the seller knows that bidders play undominated strategies.  I find that a \textit{second-price auction with a random reserve} maximizes the worst-case expected revenue over all possible signal distributions and all equilibria in undominated strategies. 
\vspace{0in}
\end{abstract}
\newpage

\section{Introduction}
The classic auction theory assumes that the seller knows bidders' information structure and derives optimal (revenue-maximizing) mechanisms (e.g., \cite{myerson1981optimal}). However,  optimal mechanisms vary  with the model of bidders' information structure and substantially less is known about how optimal auctions would perform if the model were misspecified. More importantly,   the seller may not know bidders' information structure in practice,  and  therefore it is not clear how the seller should come up with a model of the information structure.\\
\indent To address these issues, I assume that the seller does not know every aspect of bidders' information structure and evaluates a mechanism using the worst-case expected revenue over uncertainties about information structures.\\
\indent More specifically, I consider a model in which two bidders are competing for a single-unit good sold by a revenue-maximizing seller. Each bidder has two possible prior valuations towards the good, either 1 or 0. Their prior valuations are independently and identically distributed. The joint distribution of their prior valuation profile is common knowledge among the seller and the bidders. Each bidder  may observe an independently and identically distributed  signal about her prior valuation. While the information structure (that  signals are independently and identically distributed as well as the signal distribution) is common knowledge between the bidders,  the seller only knows that signals are independently and identically distributed but does not know their distributions. A signal distribution that is consistent with the  the prior valuation distribution is referred to as a \textit{possible} signal distribution. The seller seeks a mechanism from a vast class of mechanisms with the only requirement that the mechanism ``secures'' bidders' participation: there is  a message for each payoff type of each bidder that guarantees a non-negative payoff regardless of the other bidder' messages. In addition, the seller knows that bidders play undominated strategies\footnote{The assumption that bidders play undominated strategies (or ``admissible strategies'') is often considered as a reasonable assumption for an individual's ``rationality'' in the literature on decision theory and game theory. See, for example, \cite{kohlberg1986strategic}. In the literature on implementation theory, \cite{palfrey1991nash} use the concept of undominated Nash equilibrium, which is the same refinement on the set of equilibria considered in this paper. \cite{yamashita2015implementation} studies robust mechanism design problems assuming that bidders play undominated strategies. In contrast to my model, bidders may not play a Bayesian equilibrium in his model. } in a mechanism. The seller believes that an adversarial Nature chooses a possible signal distribution and an equilibrium in undominated strategies to minimize the expected revenue.\\
\indent The joint mechanism design and information design problem is not a standard zero-sum game, as a given mechanism and a given information structure need not have a unique equilibrium in undominated strategies. Moreover, an equilibrium in undominated strategies may not exist at all. I address the issues of equilibrium multiplicity and existence by using a new solution concept, with a flavour similar  to the one used in \cite{brooks2021optimal}: a \textit{maxmin solution} is a triple of a mechanism, a possible signal distribution, and an  equilibrium in undominated strategies such that, i) fixing the mechanism, the possible signal distribution and the equilibrium in undominated strategies minimize the expected revenue,  and  ii) fixing the possible signal distribution, the mechanism and the equilibrium in undominated strategies maximize the expected revenue.    Indeed, these statements remain true regardless of which equilibrium in undominated strategies is played.   The maxmin solution has an associated \textit{revenue guarantee}, which is the expected revenue in the constituent equilibrium under the constituent signal distribution. The revenue 
guarantee is both a tight lower bound on the expected revenue for the mechanism across all
possible signal distributions and all equilibria in undominated strategies and a tight upper bound on the expected revenue for the signal distribution across all mechanisms and all equilibria in undominated strategies.\\
\indent The main result (Theorem \ref{t1}) constructs a maxmin solution. First, the maxmin mechanism is a  \textit{second-price auction with a random reserve}. The distribution of the random reserve is atomless and admits a  density function everywhere on $[0,1]$. Remarkably, a simple dominant-strategy mechanism arises as a robustly optimal mechanism across all participation-securing  mechanisms, which include all dominant-strategy mechanisms as well as  nondominant-strategy mechanisms used in practice, e.g., first-price auctions and all-pay auctions. Therefore, even if nondominant-stategy mechanisms are allowed, the seller would use a dominant-strategy mechanism to maximize the revenue guarantee in my setting, which is a priori unanticipated. The result is consistent with the prevalence of second-price auctions used in practice for selling a good. Indeed, the result provides  a rationale for using a second-price auction from the perspective of robustness. Moreover, the rationale is strong, as it is established across a vast class of mechanisms that incorporates arguably any conceivable practical mechanisms. Second, the minmax signal distribution is an \textit{equal-revenue distribution}, defined by the property  of   a unit-elastic demand: in the monopoly pricing problem, the monopoly's revenue from charging any price in the support of this distribution is the same. Equal-revenue distributions are familiar in several  literatures: they  emerge endogenously  in many robust mechanism design environments and information design environment, e.g., \cite{bergemann2008pricing},  \cite{carrasco2018optimal}, \cite{zhang2021correlation}, \cite{roesler2017buyer}, \cite{condorelli2020information}, \cite{chen2020information}, etc. Finally, the constituent equilibrium is the truth-telling equilibrium in which bidders always truthfully report their true signals.\\
\indent Let me give a heuristic illustration of  the result. I start with  the minmax signal distribution.  This signal distribution has the property that  each bidder's  ``virtual value'' (this is the standard Myerson's virtual value) is zero except when the bidder observes the highest possible signal. Therefore, the seller is \textit{indifferent} between a wide range of mechanisms under this signal distribution. Indeed,  the seller is indifferent between all Bayesian incentive compatible and Bayesian individually rational mechanisms in which 1) the participation constraint is \textit{binding} for a bidder with the lowest possible signal, and 2) the good is \textit{fully} allocated to the bidder(s) with the highest possible signal(s), given that the truth-telling equilibrium is played. This is because the seller is indifferent between allocating and not allocating the good when both bidders' virtual values are zeros. \\
\indent Given that the adversarial Nature chooses the worst-case equilibrium in undominated strategies, second-price auctions would be  natural candidates for a maxmin mechanism, as the truth-telling equilibrium is the unique equilibrium in undominated strategies if the mechanism is a second-price auction (with or without a reserve). In addition, given that the adversarial Nature chooses the worst-case signal distribution, certain randomization device is expected to be employed in a maxmin mechanism as it would hedge against uncertainties over the signal distributions. Hence,  I propose a second-price auction with a random reserve as a candidate for a maxmin mechanism. The distribution of the random reserve is then constructed so that the minmax signal distribution minimizes the expected revenue in the truth-telling equilibrium across all possible signal distributions. More elaborately,  given a second-price auction with a random reserve, the adversarial Nature solves  an expected revenue minimization problem subject to the constraint that the signal distribution is possible. I construct a Lagrangian and use the first-order condition to derive a differential equation that the distribution of the random reserve  satisfies so that the minmax signal distribution is a solution to the Nature's constrained minimization problem. \\
\indent The remainder of the introduction discusses related work. Section \ref{s2} presents the model.  Section \ref{s3} characterizes the main result. Section \ref{s4} discusses and extends the main result.   Section \ref{s5} is a conclusion. 
\subsection{Related Work}
This paper lies at the intersection of several different literatures.  The first literature is the classic auction design literature, initiated by the seminal work of \cite{myerson1981optimal} who characterizes optimal auctions in the independent private value environment. \cite{cremer1988full}  characterize optimal auctions in the private value environment for generic correlation structures. Strikingly, the seller is able to extract the full surplus by carefully design side bets.\\
\indent While these results are of significant theoretical interest, they rely on the common knowledge assumption. The ``Wilson doctrine''  \citep{wilson1987game}  motivates the robust mechanism design literature that searches for economic institutions not sensitive to unrealistic assumptions about the information structure.
This paper contributes to the robust mechanism design literature. My model is indeed equivalent to the private value model in which the seller knows that bidders' valuations are independently and identically distributed and knows the mean of each bidder's valuation. In this regard, the closest related work is \cite{suzdaltsev2020distributionally}, who considers a model of auction design using exactly the same framework,  and characterizes the optimal deterministic reserve price for a second-price auction. He shows that it is optimal to set the reserve price to seller's own valuation, which is zero in my setting. In contrast, I do not place any restrictions on the mechanism except for a participation security constraint,  and characterize a maxmin mechanism for the two-bidder case. Importantly, the revenue guarantee under my proposed mechanism is strictly higher than that under his mechanism for the two-bidder case. In this sense, this paper complements his work.\\
\indent This paper is closely related to \cite{che2019distributionally}, \cite{brooks2021maxmin},  \cite{he2022correlation} and \cite{zhang2021correlation}. \cite{che2019distributionally} considers a model of auction design in the private value environment and assumes that the seller only knows the mean of bidders' valuation distribution. He characterizes a second-price auction with a random reserve as a maxmin mechanism within a class of mechanisms termed as   competitive mechanisms. Interestingly, the formats of maxmin mechanisms in both papers are second-price auctions, albeit with different random reserves. In this regard, this paper provides further support on using  second-price auctions for selling a good. In addition, we both assume that bidders play undominated strategies.  However, there are several differences. First, the seller in his model does not know how bidders' valuations are correlated, whereas the seller in mine knows that  bidders' valuations are independently and identically distributed. That is, the seller knows more in my model, and therefore the revenue guarantee in my model is an upper bound of the one in his model. Second, my solution concept is stronger in that I allow the seller to choose a mechanism from the class of all participation-securing mechanisms, which is a strict superset of the class of competitive mechanisms. \\
\indent \cite{brooks2021maxmin} consider a model of auction design in the interdependent value environment and assume that the seller knows only the mean of bidders' prior valuation. They finds, among others,  that  a proportional auction,  in which the aggregate
allocation is equal to the minimum of the sum of bidders' signal and 1,  and each bidder’s individual
allocation is proportional to their signal, is a maxmin mechanism across all participation-securing mechanisms for the symmetric case. In their model, the seller    knows neither the joint distribution of  bidders' prior valuation profile nor how bidders' signals are correlated. This is in sharp contrast to my model, in which the seller  knows the joint distribution of bidders' prior valuation profile and knows that  bidders' signals about  their prior valuations  are independently and identically distributed.  Therefore, the revenue guarantee is an upper bound of the one in their model. Indeed, they find that in the minmax information structure, bidders’ prior valuation are perfectly
correlated, so essentially a common value model emerges endogenously. Hence, our  methodologies differs. In their model, the Nature's minimization problem is a linear program using a powerful tool of Bayes correlated equilibrium in \cite{bergemann2013robust}, whereas the Nature's minimization problem is a non-linear program in mine. In addition, the solution concept in their paper is stronger than the one in mine, as there are no restrictions on the set of equilibria that the Nature can choose. \\
\indent \cite{he2022correlation} and \cite{zhang2021correlation} both consider a model of auction design in the private value environment and assume that the seller knows the marginal distribution of a generic bidder's valuation but does not knows the correlation structure between bidders' valuations. Under different conditions on the marginal distributions, they characterize, among others, that second-price auctions with random reserves are maxmin mechanisms within (standard) dominant-strategy mechanisms. One of the main differences  is that the sets over which  the seller evaluates worst-case expected revenue  are  different: the seller knows the marginal distribution but not the correlation structure in those two papers, whereas the seller knows the correlation structure but not the marginal distribution (except for the mean) in this paper. In addition, the solution concept in this paper is stronger: the class of mechanisms considered in this paper is much wider than the class of  dominant-strategy mechanisms.\\ 
\indent This paper is also related to \cite{carrasco2018optimal}, \cite{zhang2021robust} and \cite{zhang2022robust}. \cite{carrasco2018optimal} consider a model of monopoly selling in which the seller sells a good to a single buyer when the seller knows an arbitrary number of moment conditions. When there is only one bidder, my model is reduced to to their special case in which the seller knows only the mean. However, adding a second bidder with independently and identically distributed valuations, there is strictly more competition in my model. Therefore, the seller guarantees a higher revenue in my model. \cite{zhang2021robust} considers a model of bilateral trade and characterizes optimal dominant-strategy mechanisms when the profit-maximizing  intermediary knows only the mean of each trader's valuations.  \cite{zhang2022robust} considers a model of public good provision and characterizes optimal  dominant-strategy mechanisms when the profit-maximizing  principal knows only the mean of each agent's valuation. One of the main differences from those two papers is that this paper studies auction designs. Moreover, the designers in those two papers knows less: the designers do not know how agents' valuations are correlated. \\
\indent Finally, this paper is  related to the information design literature: see \cite{kamenica2019bayesian} and \cite{bergemann2019information} for recent surveys. The closest related paper in this literature is \cite{chen2020information} who study information design problems in the auction model. More specifically, they assume that bidders acquire independently and identically distributed signals. The seller, after observing the choice of the  distribution  but not the signal  realizations, designs a revenue-maximizing mechanisms.    They find, among others, that the seller-worst information structure is an equal-revenue distribution for the two-bidder case. It has been an open question  whether strong duality holds in the independent private value model. This paper provides a positive answer to (a version of) this  question for the two-bidder case. Indeed, the result in this paper implies that the seller-worst information structure is an equal-revenue distribution for the two-bidder case. In addition, when there are more than two possible prior valuations, I provide a sufficient condition such that the  result still holds (Corollary \ref{c1}), which implies that the seller-worst information structure is an equal-revenue distribution for the two-bidder case when that sufficient condition holds. That sufficient condition covers many prior valuation distributions that are not covered in their paper. In this regard, this paper complements their work. 
\section{Model}\label{s2}
\subsection{Information}
A seller sells a single unit of good.  For exposition, I assume that the supply cost of the good is zero. There are two bidders, indexed by $i\in \{1,2\}$,   competing for the good.   Each bidder’s prior valuation, denoted by $v_i$, is identically and independently drawn
from a Bernoulli distribution $F$ on $\{0,1\}$. I denote by $v=(v_1,v_2)$ the prior valuation profile.  I denote by $F\times F$ the joint distribution of the prior valuation profile. $F\times F$ is common knowledge among the seller and the bidders.  I denote  by $\mu = E[v_i ] = Pr(v_i = 1)$ the  expectation of each bidder's prior valuation. To rule
out trivial cases, I assume that $\mu \in (0,1)$.\\
\indent Each bidder $i$ may observe an independently and
identically distributed signal $s_i$ about $v_i$.   The fact that $s_1$ and $s_2$ are independently and
identically distributed and the signal distribution are  common knowledge between the two bidders. In contrast, the seller  knows that $s_1$ and $s_2$ are independently and
identically distributed but   does not know the signal distribution. In addition, bidders play undominated strategies, which is known by the seller.  Following \cite{roesler2017buyer}, I say a signal distribution is \textit{possible} if each signal of a bidder
provides her with an unbiased estimate about her valuation, or $E[v_i|s_i]=s_i$. Using Blackwell's characterization 
\citep{blackwell1953equivalent}, the prior valuation distribution $F$ is a mean-preserving spread of a possible signal distribution. I denote the set of possible signal distributions by $\mathcal{G}_F$ with a typical element $G$.  Because $F$ is a Bernoulli distribution on $\{0,1\}$, a signal distribution is possible if and only if the mean is $\mu$.  Formally, \[\mathcal{G}_F=\{G:[0,1]\to [0,1]|\int _0^1sdG(s)=\mu \quad\text{and $G$ is a CDF}\}.\] 
\subsection{Mechanism}
\indent A \textit{mechanism} $\mathcal{M}$ consists of measurable sets of messages $M_i$ for each $i$ and measurable allocation rules $q_i:M\to [0,1]$ and measurable payment rules: $t_i:M\to \mathbb{R}$ for each $i$, where $M=M_1 \times M_2$ is the set of message profiles, such that $q_1(m)+q_2(m)\le 1$ for each $m\in M$.  Bidders' preferences are quasilinear. Given a mechanism $\mathcal{M}$ and a simultaneously submitted message profile $m$, bidder $i$ with a signal of $s_i$ has a utility  \[
U_i(s_i,m)=s_i\cdot q_i(m)-t_i(m).\tag{1}\label{1}\]
 I require the mechanism to satisfy a   \textit{participation security} constraint: For each $i$, there exists a message $m_i\equiv 0\in M_i$ such that for each $s_i\in [0,1]$ and each $m_{-i}\in M_{-i}$, \[U_i(s_i,(0,m_{-i}))\ge 0.\tag{PS}\label{ps}\] Bidder $i$ with a signal $s_i$ can guarantee a non-negative utility  by sending this message, regardless of messages sent by the other bidder.
 \subsection{Equilibrium}
 \indent Given a mechanism $\mathcal{M}$ and a signal distribution $G$, I have a game of incomplete information. A \textit{Bayes Nash Equilibrium} (BNE) of the game is a strategy profile $\sigma=(\sigma_i)$, $\sigma_i:[0,1]\to \Delta(M_i)$, such that $\sigma_i$ is best response to $\sigma_{-i}$: let $U_i(s_i,\mathcal{M},G, \sigma)=\int_{s_{-i}}U_i(s_i, (\sigma_i(s_i),\sigma_{-i}(s_{-i})))dG(s_{-i})$ where $U_i(s_i, (\sigma_i(s_i),\sigma_{-i}(s_{-i})))$ is the multilinear extension of $U_i$ in Equation \eqref{1}, then for any $i,s_i,\sigma_i'$,
    \[U_i(s_i, \mathcal{M},G, \sigma)\ge U_i(s_i, \mathcal{M}, G, (\sigma_i',\sigma_{-i})).\tag{BR}\label{br} \]
The set of  all Bayes Nash Equilibria in undominated strategies for a given mechanism $\mathcal{M}$ and a given signal distribution $G$ is denoted by $\Sigma(\mathcal{M},G)$.\\
\indent Given a mechanism $\mathcal{M}$, the expected revenue at a signal distribution $G$ and an equilibrium $\sigma$ is 
\[R(\mathcal{M}, G,\sigma)=\int_{s_1}\int_{s_2}[t_1(\sigma_1(s_1),\sigma_2(s_2))+t_2(\sigma_1(s_1),\sigma_2(s_2))]dG(s_2)dG(s_1).\]
I refer to the integrand as the \textit{interim revenue} given the signal profile $(s_1,s_2)$ at the equilibrium $\sigma$.
\subsection{Solution concept}
 I adopt the following solution concept:\\
\indent A \textit{maxmin solution} is a triple $(\mathcal{M},G,\sigma)$  of a  mechanism, a possible signal distribution, and a strategy profile, with profit $R=R(\mathcal{M}, G,\sigma)$, such that the following conditions are satisfied:\\
\indent\indent C1. For any possible signal distribution $G'$ and any equilibrium in undominated strategies $\sigma'$  in $\Sigma(\mathcal{M},G')$,  $R\ge R(\mathcal{M}, G',\sigma')$.\\
\indent\indent C2. For any mechanism $\mathcal{M}'$ and any equilibrium in undominated strategies $\sigma'$  in \indent\indent\indent $\Sigma(\mathcal{M}',G)$, $R\ge R(\mathcal{M}', G,\sigma')$.\\
\indent\indent C3. Strategy profile $\sigma$ is in $\Sigma(M,G)$.\\
I refer to $R$ as the revenue guarantee of the solution. Note that if $\Sigma(\mathcal{M},G')=\emptyset$, then the condition C1 holds trivially; similarly, if $\Sigma(\mathcal{M}',G)=\emptyset$, then the condition C2 holds trivially. \\
\indent My solution concept is similar to but weaker than the ``strong maxmin solution'' in \cite{brooks2021optimal}: On the one hand,  we both allow general participation-securing mechanisms. On the other hand,  I restrict attention to equilibria in undominated strategies, whereas they do not place any restrictions on  equilibria. 
\section{Main Result}\label{s3}
In this section, I first formally define a mechanism $\bar{\mathcal{M}}$ (Section \ref{s31}),  a possible signal distribution  $\bar{G}$ (Section \ref{s32}) and a strategy profile $\bar{\sigma}$ (Section \ref{s33}), then I present the formal statement of the result (Section \ref{s34}) that the proposed mechanism, the proposed signal distribution together with the proposed strategy profile constitute a maxmin solution. Finally,  I prove the formal statement (Section \ref{s35}).
\subsection{Mechanism $\bar{\mathcal{M}}$}\label{s31}
The mechanism $\bar{\mathcal{M}}$ is a second-price auction with a random reserve whose cumulative distribution function $\bar{H}$ is as follows:
\[\bar{H}(x)=\left\{
\begin{array}{lll}
-\frac{x(1-a)(\ln{x}-\ln{a})}{(x-a)\ln{a}}   &      & {\text{if $x\in (0,a)\cup (a,1]$},}\\
-\frac{1-a}{\ln{a}}    &      & {\text{if $x=a$},}\\
0    &      & {\text{if $x=0$},}
\end{array} \right.\]
where $a\in (0,1)$ is the unique solution to \[\tilde{a}(1-\ln{\tilde{a}})=\mu. \tag{2}\label{2}\]
\indent To see that $\bar{H}$ is a distribution on $[0,1]$, note first that $\bar{H}(x)$ is continuous, following from $\lim_{x\to a}\bar{H}(x)=\bar{H}(a)$ and $\lim_{x\to 0}\bar{H}(x)=\bar{H}(0)=0$ using L'H\^opital's rule. Second, it can be shown that $\bar{H}(x)$ is strictly increasing (Lemma \ref{l1}). Third, $\bar{H}(1)=1$.  In addition, it is straightforward to show that $\bar{H}(x)$ is differentiable at any $x\neq 0$ using L'H\^opital's rule, i.e., $\bar{H}$ is a continuous distribution (Lemma \ref{l1}).   See Figure \ref{fig:S1} for an illustration.
\begin{figure}
\centering
\begin{tikzpicture}
\begin{axis}[
    axis lines = left,
    xmin=0,
        xmax=1.3,
        ymin=0,
        ymax=1.3,
        xtick={0,1,1.3},
        ytick={0,1,1.3},
        xticklabels = {$0$,  $1$, $x$},
        yticklabels = {$0$,  $1$, $\bar{H}(x)$},
        legend style={at={(1.1,1)}}
]

\addplot[domain=0:1,color=black, name path=A] {x^0.5};
\addplot[dashed] coordinates {(0, 1) (1,1)};
\addplot[dashed] coordinates {(1, 0) (1,1)};
\end{axis}
\end{tikzpicture}
\caption{CDF of the Random Reserve in $\bar{\mathcal{M}}$} \label{fig:S1}
\end{figure}
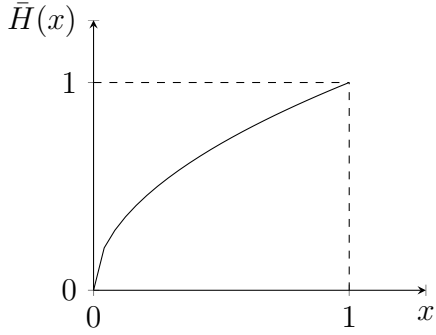
\begin{lemma}\label{l1}
$\bar{H}(x)$ is strictly increasing. In addition, $\bar{H}(x)$ is differentiable at any $x\neq 0$. Moreover, $\lim_{x\to 0}x\bar{H}'(x)=0$.
\end{lemma}
\begin{proof}
For any $x\neq 0$ or $a$, $\bar{H}'(x)=-\frac{(1-a)(x-a\ln{x}-\mu)}{(x-a)^2\ln{a}}$. Define $J(x)\equiv x-a\ln{x}$. Because $J'(x)=1-\frac{a}{x}$ and $J''(x)=\frac{a}{x^2}>0$, $J(x)$ is minimized at $x=a$ and the minimized value is equal to $a-a\ln{a}=\mu$. Therefore for any $x\neq 0$ or $a$, $\bar{H}'(x)>0$. In addition, $\lim_{x\to a}\bar{H}'(x)=-\frac{1-a}{2a\ln{a}}>0$ using L'H\^opital's rule (twice). Moreover, $\lim_{x\to 0}x\bar{H}'(x)=-\frac{(1-a)(x^2-ax\ln{x}-x\mu)}{(x-a)^2\ln{a}}=0$ using L'H\^opital's rule.
\end{proof}
\indent More formally, the mechanism  $\bar{\mathcal{M}}=(\bar{M},\bar{q}_i,\bar{t}_i)_{i\in\{1,2\}}$ is defined as follows. It is a direct mechanism, i.e., $\bar{M}=[0,1]^2$. With slight abuse of notations, I  denote by  $(s_1,s_2)$ the reported message profile. If $s_1>s_2$, then $\bar{q}_1(s_1,s_2)=\bar{H}(s_1)$, $\bar{q}_2(s_1,s_2)=0$,  $\bar{t}_1(s_1,s_2)=s_1\bar{H}(s_1)-\int_{s_2}^{s_1}\bar{H}(x)dx$, $\bar{t}_2(s_1,s_2)=0$; if $s_1<s_2$, then $\bar{q}_1(s_1,s_2)=0$, $\bar{q}_2(s_1,s_2)=\bar{H}(s_2)$, $\bar{t}_1(s_1,s_2)=0, \bar{t}_2(s_1,s_2)=s_2\bar{H}(s_2)-\int_{s_1}^{s_2}\bar{H}(x)dx$; if $s_1=s_2=x$, then $\bar{q}_1(s_1,s_2)=\bar{q}_2(s_1,s_2)=\frac{\bar{H}(x)}{2}$,  $\bar{t}_1(s_1,s_2)=\bar{t}_2(s_1,s_2)=\frac{x\bar{H}(x)}{2}$.\\
\subsection{Signal distribution $\bar{G}$}\label{s32}
$\bar{G}$ is an equal-revenue distribution as follows:
\[\bar{G}(x)=\left\{
\begin{array}{lll}
1-\frac{a}{x}   &      & {\text{if $a\le x<1$},}\\
1   &      & {\text{if $x=1$}.}
\end{array} \right.\]
Note that $\bar{G}$ is a possible signal distribution because $a$ is a solution to Equation \eqref{2}. Define $\bar{R}\equiv 2a-a^2$. As will be shown (Section \ref{s352}), $\bar{R}$ is an upper bound of the expected revenue for any participation-securing mechanism and any equilibrium under the signal distribution $\bar{G}$. 
\subsection{Strategy profile $\bar{\sigma}$}\label{s33}
Finally, let $\bar{\sigma}$ be the truth-telling strategy profile in the mechanism $\bar{\mathcal{M}}$ under the signal distribution $\bar{G}$: for all $i$ and $s_i$, $\bar{\sigma}_i(s_i)$ puts probability one on $s_i$. This completes the construction of the solution. 
\subsection{Formal statement: Theorem \ref{t1} }\label{s34}
\begin{theorem}\label{t1}
 $(\bar{\mathcal{M}},\bar{G},\bar{\sigma})$ is a maxmin solution with a revenue guarantee of $\bar{R}$.
\end{theorem}
\subsection{Proof of Theorem \ref{t1}}\label{s35}
\subsubsection{Lower Bound on Revenue for $\bar{\mathcal{M}}$}\label{s351}
I first establish C1 in the definition of a maxmin solution.
\begin{proposition}\label{p1}
For any possible signal distribution $G$ and any equilibrium in undominated strategies $\sigma$ in $\Sigma(\bar{\mathcal{M}},G)$, $R(\bar{\mathcal{M}},G,\sigma)\ge \bar{R}$.
\end{proposition}
\begin{proof}
Given the mechanism $\bar{\mathcal{M}}$, truth-telling is the unique equilibrium in undominated strategies under any signal distribution. Therefore, 
I  focus on the truth-telling equilibrium and show that $\bar{G}$ minimizes the expected revenue across possible signal distributions. Given a signal profile $(s_1,s_2)$,  let $s(1)$ be $\max\{s_1,s_2\}$  and $s(2)$ be $\min\{s_1,s_2\}$. Then in the truth-telling equilibrium, the expected revenue given the mechanism $\bar{\mathcal{M}}$ and an arbitrary signal distribution $G$ can be expressed as follows:
\[\begin{split}
    E[t_1(s_1,s_2)+t_2(s_1,s_2)] & =\int_0^1\int_0^1 [s(1)\bar{H}(s(1))-\int_{s(2)}^{s(1)}\bar{H}(x)dx]dG(s_1)dG(s_2)\\
    &= \int_0^1\int_0^1 s(1)\bar{H}(s(1))dG(s_1)dG(s_2)-\int_0^1\int_0^1\int_0^1\bar{H}(x)\mathbbm{1}_{s(2)\le x \le s(1)}dxdG(s_1)dG(s_2)\\
    &= \int_0^1 \{\underbrace{(1-G^2(x))[x\bar{H}'(x)+\bar{H}(x)]}_{\text{the first term}}-\underbrace{\bar{H}(x)[2G(x)(1-G(x))]}_{\text{the second term}}\}dx,
\end{split}\]
where the first term  of the last line is obtained using integration by parts\footnote{$x\bar{H}(x)$ is differentiable everywhere by Lemma \ref{l1}.} and  the fact that $s(1)$'s cumulative distribution function is $G^2$, and the second term of the last line is obtained using Fubini's theorem and the fact that $s(2)$'s cumulative distribution function is $G^2+2G(1-G)$. Note that $\int_0^1(1-G(x))dx=\mu$ holds for any possible signal distribution, using integration by parts. Then to show that $\bar{G}$ minimizes $\mathcal{L}(G,\bar{H})\equiv \int_0^1 \{(1-G^2(x))[x\bar{H}'(x)+\bar{H}(x)]-\bar{H}(x)[2G(x)(1-G(x))]\}dx$, it suffices to show that there exists a real number $\lambda$ such that $\bar{G}$ minimizes \[\mathcal{L}(G,\bar{H},\lambda)\equiv\int_0^1 \{(1-G^2(x))[x\bar{H}'(x)+\bar{H}(x)]-\bar{H}(x)[2G(x)(1-G(x))]-\lambda(1-G(x))\}dx.\] This is because $\mathcal{L}(\bar{G},\bar{H},\lambda)\le \mathcal{L}(G,\bar{H},\lambda)$ implies $\mathcal{L}(\bar{G},\bar{H})\le \mathcal{L}(G,\bar{H})$ for any possible signal distribution $G$ by adding $\lambda\mu$ to both sides. Now take $\lambda=-\frac{2(1-a)}{\ln{a}}$. I am going to point-wise minimize the integrand of $\mathcal{L}(G,\bar{H},\lambda)$ for any $x\neq 0,a$ or 1.   First,  with slight rewriting,  the integrand of $\mathcal{L}(G,\bar{H},\lambda)$ becomes \[
\mathcal{I}(G,\bar{H},\lambda)\equiv [\bar{H}(x)-x\bar{H}'(x)]G^2(x)-2[\bar{H}(x)+\frac{1-a}{\ln{a}}]G(x)+\bar{H}(x)+x\bar{H}'(x)+\frac{2(1-a)}{\ln{a}}.\tag{3} \label{3}\]
This is a simple quadratic function of $G(x)$. By simple calculation, $\bar{H}(x)-x\bar{H}'(x)=\frac{x[\bar{H}(x)+\frac{1-a}{\ln{a}}]}{x-a}>0$ for any $x\neq 0$ or $a$, following from $\bar{H}(x)$ being strictly increasing and $\bar{H}(a)=-\frac{1-a}{\ln{a}}$. Then,  for any $a<x<1$, $G(x)=1-\frac{a}{x}$ is the unique minimizer of \eqref{3}. For any $0<x<a$, $G(x)=0$ is the unique minimizer of \eqref{3} as $1-\frac{a}{x}<0$ if $0<x<a$.  This implies that  the signal distribution $\bar{G}$ minimizes $\mathcal{L}(G,\bar{H},\lambda)$ and therefore minimizes the expected revenue across possible signal distributions. By simple calculation, the minimized expected revenue is equal to $2a-a^2$. The details about the construction of the distribution $\bar{H}$ as well as the Lagrangian multiplier $\lambda$ are provided below
\end{proof}
\textit{\large \textbf{Construction of $\bar{H}$ and $\lambda$}}. Consider a second-price auction with a random reserve whose cumulative distribution function is $H$. Assuming that $H(x)$ is differentiable at any $x\in (0,1]$ and $\lim_{x\to 0}xH(x)$ exists, then the expected revenue under a signal distribution $G$ can be expressed as $\mathcal{L}(G,H)$. I subtract $\lambda\mu$ from $\mathcal{L}(G,H)$ where $\lambda$ is some real number,  and obtain $\mathcal{L}(G,H,\lambda)$. Then a sufficient condition for $\bar{G}$ to be a minimizer of $\mathcal{L}(G,H)$ is that $\bar{G}(x)$ point-wise minimizes the integrand of $\mathcal{L}(G,H,\lambda)$. The first order condition with respect to $G(x)$ is as follows:
\[-2G(x)[xH'(x)+H(x)]-2H(x)(1-2G(x))+\lambda=0.\tag{4}\label{b1}\]
Plugging $G(x)=1-\frac{a}{x}$ to \eqref{b1}, 
\[(x-a)H'(x)+\frac{a}{x}\cdot H(x)=\frac{\lambda}{2}. \tag{5}\label{5}\]
Solving this differential equation: for any $x\neq 0$ or $a$, \[
H(x)=\frac{x(\frac{\lambda}{2}\ln{x}+c)}{x-a}, \tag{6}\label{b3}\]
where $c$ is some constant. Using $H(1)=1$, I obtain that $c=1-a$. In order for $\lim_{x\to a}H(a)$ to exist, the nominator of \eqref{b3} has to be 0 when $x=a$, leading to that $\lambda=-\frac{2(1-a)}{\ln{a}}$.
\begin{remark}
\normalfont The assumption that bidders play undominated strategies is important for this result. Without this assumption, there indeed exist a possible signal distribution and an equilibrium such that the expected revenue in the mechanism $\bar{\mathcal{M}}$ is lower than $\bar{R}$. Take the signal distribution to be a point mass on $\mu$, i.e., bidders observe no signal at all, and consider the strategy profile where one bidder reports $\mu$ and the other bidder reports 0. It is straightforward to verify that this strategy profile  is an equilibrium (in which one of the bidders uses a dominated strategy). The expected revenue in this equilibrium under this signal distribution is $\mu \cdot \bar{H}(\mu)-\int_0^{\mu}\bar{H}(x)dx$, which is lower than $2a-a^2$ for any $\mu$. For a parametric example, when $\mu=0.5$, the former is about 0.1223, whereas the latter is about 0.3385.
\end{remark}
\subsubsection{Upper Bound on Revenue for $\bar{G}$}\label{s352}
Next,  I  establish C2 in the definition of a maxmin solution.
\begin{proposition}\label{p2}
For any participation-securing mechanism $\mathcal{M}$ and any equilibrium in undominated strategies $\sigma$ in $\Sigma(\mathcal{M},\bar{G})$,  $R(\mathcal{M},\bar{G},\sigma)\le \bar{R}$.
\end{proposition}
\begin{proof}
Indeed, I will establish a slightly stronger statement: $\bar{R}$ is an upper bound on the expected revenue for any participation-securing mechanism and any equilibrium. \\
\indent  First, to identify an upper bound on the expected given a signal distribution, it is without loss to restrict attention to direct mechanisms, i.e., $M=[0,1]^2$, as the revelation principle holds. Next, for exposition, I parameterize each type $s_i$ by its quantile $z_i$.\footnote{See \cite{carroll2017robustness} and \cite{zhang2022auctioning} for a similar method.} Formally, I define the inverse quantile function as follows: \[\bar{s}(z_i)=\min\{\tilde{s_i}|\bar{G}(\tilde{s_i})\ge z_i\}=\left\{
\begin{array}{lll}
\frac{a}{1-z_i}    &      & {\text{if $0\le z_i<1-a$},}\\
1    &      & {\text{if $z_i\ge 1-a$}.}
\end{array} \right.\]
I denote the signal profile by $z=(z_1,z_2)\in [0,1]^2$. Note that $z_1$ and $z_2$ follow independently  and identically  distributed uniform distributions on $[0,1]$. Then, for any direct mechanism $(q_i(z),t_i(z))_{i\in\{1,2\}}$ where $q_i(z)\in [0,1]$ is the allocation probability to  bidder $i$ given the parameterized signal profile $z$ and $t_i(z)\in \mathrm{R}$ is the payment made by bidder $i$ given $z$,  \eqref{br} and \eqref{ps} together imply that for all $i$, all $z_i$, and all $z_i'$,\[U_i(z_i)\equiv \bar{s}(z_i)Q_i(z_i)-T_i(z_i)\ge \bar{s}(z_i)Q_i(z_i')-T_i(z_i'), \tag{BIC}\label{bic}\]
\[\bar{s}(z_i)Q_i(z_i)-T_i(z_i)\ge 0, \tag{BIR}\label{bir}\]
where $Q_i(z_i)=\int_{z_{-i}}q_i(z_i,z_{-i})dz_{-i}$ and $T_i(z_i)=\int_{z_{-i}}t_i(z_i,z_{-i})dz_{-i}$ are the expected allocation to type $z_i$ of bidder $i$ and the expected payment made by type $z_i$ of bidder $i$,  respectively.\\ \indent For $z_i'\ge z_i$,  \eqref{bic} implies that  \[
(\bar{s}(z_i')-\bar{s}(z_i))Q_i(z_i')\ge U_i(z_i')-U_i(z_i)\ge (\bar{s}(z_i')-\bar{s}(z_i))Q_i(z_i).\tag{2}\label{4}\]
Then  $U_i(z_i)$ is Lipschitz, thus absolutely continuous w.r.t. $z_i$, and so equal to the integral of its
derivative. In addition, note that $\bar{s}(z_i)$ is differentiable for all $z_i$ except for  $z_i=1-a$. Then applying the envelope theorem to \eqref{4} at each point of differentiability, I obtain that 
\[\frac{\partial U_i(z_i)}{\partial z_i}=\frac{\partial \bar{s}(z_i)}{\partial z_i}Q_i(z_i)=\left\{
\begin{array}{lll}
\frac{a}{(1-z_i)^2}Q_i(z_i)    &      & {\text{if $0\le z_i<1-a$},}\\
0    &      & {\text{if $z_i> 1-a$}.}
\end{array} \right.\]
Thus, \[U_i(z_i)=\left\{
\begin{array}{lll}
U_i(0)+\int_0^{z_i}[ \frac{a}{(1-\tilde{z_i})^2}Q_i(\tilde{z_i})] d\tilde{z_i}   &      & {\text{if $0\le z_i<1-a$},}\\
U_i(0)+\int_0^{1-a}[ \frac{a}{(1-\tilde{z_i})^2}Q_i(\tilde{z_i})] d\tilde{z_i}  &      & {\text{if $z_i\ge 1-a$}.}
\end{array} \right.\]
Therefore,   the expected revenue from bidder $i$  satisfies 
\[
\begin{split}
\int_0^1T_i(z_i)dz_i & = \int_0^1[\bar{s}(z_i)Q_i(z_i)-U_i(z_i)]dz_i \\
 & = \int_0^{1-a}\{\bar{s}(z_i)Q_i(z_i)-U_i(0)-\int_0^{z_i}[\frac{a}{(1-\tilde{z_i})^2}Q_i(\tilde{z_i})] d\tilde{z_i}\}dz_i+\\
 &\int_{1-a}^1\{\bar{s}(z_i)Q_i(z_i)-U_i(0)-\int_0^{1-a}[\frac{a}{(1-\tilde{z_i})^2}Q_i(\tilde{z_i})] d\tilde{z_i}\}dz_i\\
 &\le \int_0^{1-a}\{\bar{s}(z_i)Q_i(z_i)-\int_0^{z_i}[\frac{a}{(1-\tilde{z_i})^2}Q_i(\tilde{z_i})] d\tilde{z_i}\}dz_i+\\
 &\int_{1-a}^1\{\bar{s}(z_i)Q_i(z_i)-\int_0^{1-a}[ \frac{a}{(1-\tilde{z_i})^2}Q_i(\tilde{z_i})] d\tilde{z_i}\}dz_i\\
 &=\{\int_0^{1-a}[(\bar{s}(z_i)-(1-a-z_i)\frac{a}{(1-z_i)^2})Q_i(z_i)]dz_i+\\
 &\int_{1-a}^1[\bar{s}Q_i(z_i)-\int_0^{1-\frac{1}{e}}[\frac{a}{(1-\tilde{z_i})^2}Q_i(\tilde{z_i})] d\tilde{z_i}]dz_i\}\\
 &=\{\int_0^{1-a}[(\bar{s}(z_i)-(1-z_i)\frac{a}{(1-z_i)^2})Q_i(z_i)]dz_i+\int_{1-a}^1[(\bar{s}(z_i)Q_i(z_i)]dz_i\}\\
 &=\int_{1-a}^1Q_i(z_i)dz_i,
\end{split}
\]
where the first inequality holds because (\ref{bir}) implies that $U_i(0)\ge 0$, the third equality is obtained via integration by parts, the last equality holds because $\bar{s}(z_i)-(1-z_i)\frac{a}{(1-z_i)^2}=0$ for $0\le z_i< 1-a$ and $\bar{s}(z_i)=1$ for $z_i>1-a$.\\
\indent Then, the expected revenue from all the bidders satisfies
\[
\begin{split}
 \sum_{i=1}^2\int_0^1T_i(z_i)dz_i & =\int_{0}^{1-a}\int_{1-a}^1 q_1(z_1,z_2)dz_1dz_2+\int_{0}^{1-a}\int_{1-a}^1 q_2(z_1,z_2)dz_2dz_1\\ &+\int_{1-a}^{1}\int_{1-a}^{1} [q_1(z_1,z_2)+q_2(z_1,z_2)]dz_1dz_2 \\ & \le 2a(1-a)+a^2=2a-a^2,   
\end{split}
\]
where the  inequality holds because $q_i(z)\le 1$ and  $q_1(z)+q_2(z)\le 1$ for all $i$ and $z$. This finishes the proof.
\end{proof}
\indent This argument is  standard  \citep{myerson1981optimal}. The parameterization makes the proof clean. \\
\indent Finally, recall that the truth-telling strategy profile $\bar{\sigma}$ is an equilibrium in undominated strategies in the mechanism $\bar{\mathcal{M}}$ under the signal distribution $\bar{G}$, as the mechanism $\bar{\mathcal{M}}$ is a dominant-strategy mechanism. Then,  Theorem \ref{t1} follows immediately from Proposition \ref{p1} and \ref{p2}.
\section{Discussion and Extension}\label{s4}
\subsection{Value of randomization}
Recall that \cite{suzdaltsev2020distributionally} characterizes the optimal deterministic reserve price for the second-price auction in the same framework, which is zero in my setting. For the two-bidder case, my proposed mechanism, which involves randomization,  achieves a strictly higher revenue guarantee. For a parametric example, if $\mu=0.5$, then the revenue guarantee under my proposed mechanism is about  0.3385, whereas the one under his mechanism is 0.25. Intuitively, randomization hedges against uncertainty towards signal distributions, rendering a higher revenue guarantee. To my knowledge, it is an open question what the optimal deterministic mechanism is in this setting. Therefore, the difference between the  revenue guarantee of my mechanism  and the one of his can be interpreted as an upper bound on the ``value of randomization''. 
\subsection{Cost of correlation}
Recall that \cite{che2019distributionally} finds that a second-price auction with a random reserve maximizes the revenue guarantee across a wide range of mechanisms in the private value environment in which the seller only knows the expectation of the value distribution. That is, values across bidders can be correlated. As expected, the revenue guarantee of my proposed mechanism is strictly higher in my setting for the two-bidder case than the one in his setting. This is because the seller in my setting knows more: he knows that the correlation structure is the independent one. For a parametric example, if $\mu=0.5$,  then the revenue guarantee of his mechanism  in his setting is about 0.317, which is strictly smaller than 0.3385. Intuitively, independently and identically distributed signal distributions makes my setting more competitive. Indeed, in his worst-case value distribution, the competitor of the high-valuation bidder always has the lowest possible valuation. To my knowledge, it is not known what the maxmin mechanism is across all participation-securing mechanisms (even across all dominant-strategy mechanisms) in his setting. Therefore, the difference between the  revenue guarantee in my paper  and the one in his paper can be interpreted as an upper bound on the ``cost of correlation''. 
\subsection{Other maxmin solutions}
\begin{corollary}
Let $\bar{\mathcal{M}}^*$ be a second price auction with a random reserve whose cumulative distribution is  $H^*$. Then $(\bar{\mathcal{M}}^*, \bar{G}, \bar{\sigma})$ is a maxmin solution with a revenue guarantee of $\bar{R}$ if $H^*$ satisfies the following properties:\\
P1. $H^*(x)=\bar{H}(x)$ for $x\ge a$.\\
P2. $H^*(x)-x\cdot(H^*)'(x)\ge 0$ for $x<a$.
\end{corollary}
\begin{proof}
By the proof of Proposition \ref{p1}, the property $P1$ implies that for any $a<x<1$, $G(x)=1-\frac{a}{x}$ is the unique minimizer of $\mathcal{I}(G, H^*,\lambda)$. Moreover,   the properties $P1$ and $P2$ together  imply that for any $0<x<a$, $G(x)=0$ is the unique minimizer of $\mathcal{I}(G, H^*,\lambda)$,  as $\mathcal{I}(G, H^*,\lambda)$ is a (weakly) convex function of $G$.
\end{proof}
\indent In a maxmin solution,  there is some flexibility for the distribution of a random reserve when the reserve is below $a$.   One example that satisfies $P2$ is that $H^*(x)=\bar{H}(a)$ for $x<a$, i.e., there is a probability mass of size $\bar{H}(a)$ on zero. 
\subsection{Knowing the second moment}
Suppose that the seller knows only the second moment of the signal distribution as well as that the two signals are independently and identically distributed. Let $\delta$ denote the known second moment, i.e., $\delta=\int_0^1x^2dG(x)$. Then I will show that the second-price auction with the uniformly distributed random reserve ($\hat{\mathcal{M}}$), the equal-revenue distribution with $a=1-\sqrt{1-\delta}$ ($\hat{G}$), and the truth-telling strategy profile ($\hat{\sigma}$) constitute a maxmin solution in this case.
\begin{corollary}
 If the seller knows only the second moment of the signal distribution as well as that the two signals are independently and identically distributed, then $(\hat{\mathcal{M}},\hat{G},\hat{\sigma})$ is a maxmin solution with a revenue guarantee of $\delta$.
\end{corollary}
\begin{proof}
It suffices to show that $\hat{G}$ minimizes the expected revene in the truth-telling equilibrium under the mechanism $\hat{\mathcal{M}}$. First, 
using integration by parts, the constraint  $\delta=\int_0^1x^2dG(x)$ can be rewritten as $\delta=2\int_0^1x(1-G(x))dx$. Similar to the proof of Proposition \ref{p1}, I construct a Lagrangian as follows:\[\hat{\mathcal{L}}(G,\hat{H},\lambda)\equiv \int_0^1 \{(1-G^2(x))[x\hat{H}'(x)+\hat{H}(x)]-\hat{H}(x)[2G(x)(1-G(x))]-2\lambda x(1-G(x))\}dx. \]
Let $\lambda$ be 1. It is straightforward that the integrand of $\hat{\mathcal{L}}(G,\hat{H},\lambda)$ is a constant of  0 because $\hat{H}(x)=x$, implying that any signal distribution with the known second moment yields the same expected revenue in the truth-telling equilibrium under the mechanism $\hat{\mathcal{M}}$. Alternatively, 
 observe  that the interim revenue at the signal profile $(s_1,s_2)$ is $\frac{s_1^2+s_2^2}{2}$ in the truth-telling equilibrium under the mechanism $\hat{\mathcal{M}}$. This also implies that the expected revenue is $\delta$ for any signal distribution with the known second moment. 
\end{proof}
\subsection{Beyond Binary Prior Valuations}
Suppose that there are more than two possible prior valuations. Recall that $G$ is a possible signal distribution if and only if $F$ is a mean-preserving-spread of $G$ \citep{blackwell1953equivalent}. Formally, the set of possible signal distributions can now be defined as follows:
\[\mathcal{G}_F=\{G:[0,1]\to [0,1]|\int_0^xF(s)ds\ge \int_0^xG(s)ds,\forall x\in [0,1], \int _0^1sdG(s)=\mu \quad\text{and $G$ is a CDF}\}.\]
\begin{corollary}\label{c1}
If $\int_0^xF(s)ds\ge \int_0^x\bar{G}(s)ds$ for all $x\in [0,1]$, then $(\bar{\mathcal{M}},\bar{G},\bar{\sigma})$ is a maxmin solution with a revenue guarantee of $\bar{R}$.
\end{corollary}
\begin{proof}
Note that this condition guarantees that $\bar{G}$ is a possible signal distribution. Then Corollary \ref{c1} follows immediately from Theorem \ref{t1}.
\end{proof}
\indent To illustrate, consider an example in which $\mu=\frac{1}{2}$ and there are three possible prior valuations: $0$, $\frac{1}{2}$ and 1. Their probabilities are $b$, $1-2b$ and $b$, respectively. By simple calculation, \[\int_0^x\bar{G}(s)ds=\left\{
\begin{array}{lll}
x-a-a\ln{x}+a\ln{a}   &      & {\text{if $x\ge a$},}\\
0    &      & {\text{if $x< a$},}
\end{array} \right.\]
where $a(1-\ln{a})=\frac{1}{2}$ in this example.
\[\int_0^xF(s)ds=\left\{
\begin{array}{lll}
bx   &      & {\text{if $x\le \frac{1}{2}$},}\\
\frac{1}{2}b+(x-\frac{1}{2})(1-b)    &      & {\text{if $x> \frac{1}{2}$}.}
\end{array} \right.\]
It is straightforward to show that if $0.2588\approx -2a\ln{\frac{1}{2}}\le b\le\frac{1}{2}$, this distribution is a mean-preserving spread of $\bar{G}$.\\
\indent For another example, suppose $\mu= \frac{3}{4}$. Consider a prior valuation distribution  which is a combination of a uniform distribution on $[0,1)$ and an atom of size $\frac{1}{2}$ on 1. Then, it can be shown that $\int_0^x\bar{G}(s)ds- \int_0^xF(s)ds$ is decreasing in $x$ if $a\le x\le  1-\sqrt{1-2a}\approx 0.515$, and is increasing in $x$ for $1-\sqrt{1-2a}<x\le 1$. Here $a(1-\ln{a})=\frac{3}{4}$.  This implies that $\int_0^x\bar{G}(s)ds- \int_0^xF(s)ds\le 0$ for any $x\in [0,1]$, as $\int_0^1\bar{G}(s)ds- \int_0^1F(s)ds= 0$.  Therefore,  this distribution is also a mean-preserving spread of $\bar{G}$. 
\section{Concluding Remarks}\label{s5}
In this paper, I find that a second-price auction with a random reserve is a maxmin mechanism across all participation-securing mechanisms for the two-bidder case. The key step of the result is the construction of a saddle point, which implies that (a version of) strong duality holds for two-bidder case in my setting. It remains an open question that if strong duality holds for more-than-two-bidder cases. This paper provides the first step towards a broad study of robust auction design problems in the independent private value model.

\bibliographystyle{apalike}
\bibliography{abc}

\end{document}